\documentclass[final,12pt]{article}
\usepackage{amsfonts,color,morefloats,pslatex,a4wide}
\usepackage{amssymb,amsthm,amsmath,latexsym,pslatex,cite}
\usepackage{lineno,hyperref,mathtools}
\usepackage{mathrsfs,amsbsy}
\usepackage{pdflscape}
\usepackage{threeparttable}
\usepackage{bm}
\usepackage{enumitem}

\newtheorem{theorem}{Theorem}[section]

\newtheorem{lemma}[theorem]{Lemma}
\newtheorem{corollary}[theorem]{Corollary}
\newtheorem{proposition}[theorem]{Proposition}

\theoremstyle{definition}
\newtheorem{definition}[theorem]{Definition}

\theoremstyle{remark}
\newtheorem{remark}[theorem]{Remark}

\makeatletter

\newcommand{\Rmnum}[1]{\expandafter\@slowromancap\romannumeral #1@}
\makeatother

\begin{document}
\date{} 
\begin{sloppypar}

\title{Improved bounds for constant-power and low-power error-correcting cooling codes
}

\author{Tingting Tong\thanks{Tingting Tong and Sihuang Hu are with the State Key Laboratory of Cryptography and Digital Economy Security, the Key Laboratory of Cryptologic Technology and Information Security, Ministry of Education, and the School of Cyber Science and Technology, Shandong University, Qingdao, Shandong 266237, China. (email: tingtingtong2000@163.com, husihuang@sdu.edu.cn).},~
Sihuang Hu}
\maketitle

\begin{abstract}
Low-power error-correcting cooling (LPECC) codes and constant-power error-correcting cooling (CPECC) codes provide error correction while controlling power consumption and thermal effects in on-chip buses. In this paper, we study binary CPECC and LPECC codes with \(e=w-3\). For CPECC codes, we extend the applicability of the upper bound previously obtained by Zhao and Zhang from the quadratic-order condition \(w\ge 2t(t+1)+2\) to \(w\ge w_0(t)\), where \(w_0(t)\sim \sqrt{2}\,t^{3/2}\). Using Steiner systems, we show that the CPECC bound is attainable and asymptotically tight for fixed \(t,w\).
For LPECC codes, we establish the new upper bound
\(\left\lfloor\frac{\binom{n+2}{3}}{\binom{w+t}{3}}\right\rfloor\)
for \(w\ge \mu(t)\), where \(\mu(t)\sim \sqrt{2}\,t^{3/2}\). This bound is strictly smaller than the previous bound of Zhao and Zhang whenever both apply, and is asymptotically tight for fixed \(t,w\) in the stated range.
\end{abstract}


\section{Introduction}\label{sec-introduction}

Power and heat dissipation have become fundamental constraints in the design of modern integrated circuits. In an on-chip or off-chip bus, a state transition on a wire causes switching activity and Joule heating, thereby contributing to dynamic power consumption. Numerous encoding techniques have been proposed to reduce bus power consumption \cite{Benini1997,Calimera2008,CheeColbournLing2006,Petrov2004,Sotiriadis2002,Sotiriadis2003,Stan1995,Wang2007,Macii1998,Sithambaram2007,Sotiriadis2001,SotiriadisWang2000,Sundaresan2005}.
However, controlling the total switching activity alone does not prevent a small number of wires from switching repeatedly and remaining significantly hotter than the others. This motivates coding schemes that control peak temperature as well as power consumption.

To address this issue, Chee et al.~\cite{2018} introduced cooling codes, in which each message is represented by a collection of
possible transition vectors so that the encoder can avoid transitions on the currently hottest wires. Low-power cooling (LPC) codes additionally restrict the number of transitions in each transmission, while low-power error-correcting cooling (LPECC) codes further incorporate error-correction capability. 
More precisely, an $(n,t,w,e)$-LPECC code is a coding scheme for communication over a bus consisting of $n$ wires satisfying the following properties:
\begin{itemize}
 \item \textbf{Property $A(t)$:} every transmission does not cause state transitions on the $t$ hottest wires;

 \item \textbf{Property $B(w)$:} the total number of state transitions on all the wires is at most $w$ in every transmission;

 \item \textbf{Property $C(e)$:} up to $e$ transmission errors (0 received as 1, or 1 received as 0) on the $n$ wires can be corrected.
\end{itemize}
A constant-power error-correcting cooling (CPECC) code instead satisfies \textbf{Property $B'(w)$}, which requires exactly $w$ state transitions in each transmission. The size of a code is the number of its information classes. We denote the maximum sizes of binary $(n,t,w,e)$-LPECC and $(n,t,w,e)$-CPECC codes by $C(n,t,w,e)$ and $C'(n,t,w,e)$, respectively.

Several constructions and bounds have been developed for cooling codes and their variants. Chee et al.~\cite{2018} introduced cooling, low-power cooling, and error-correcting cooling codes, with constructions based on complete hypergraphs, dual codes, and sunflowers. They subsequently~\cite{2020} developed efficient encoding and decoding methods for LPC codes and introduced constant-power cooling codes. Liu and Ji~\cite{2024} used packings and resolvable designs to derive upper bounds and construct optimal LPECC codes. Zhao and Zhang~\cite{2025} obtained finite-length and asymptotic bounds for LPECC and CPECC codes with large error-correcting capability. For the case where \(e=w-2\), Liu et al.~\cite{Liu2026bounds} showed new optimal \((n,t, 4, 2)\)-CPECC codes and proposed a conjecture on general \((n,t,w,w-2)\)-CPECC codes. This conjecture was resolved by Liu et al.~\cite{LiuTongZhangZhao2026}, who also obtained new bounds and optimal constructions for CPECC and LPECC codes.

These developments motivate a closer study of the next case, $e=w-3$. For this case, Zhao and Zhang~\cite[Theorem~9]{2025} established the CPECC upper bound
\[ C'(n,t,w,w-3) \le  \left\lfloor
\frac{\binom{n}{3}}{\binom{w+t}{3}}\right\rfloor.\]
For LPECC codes, they obtained the bound~\cite[Theorem~7]{2025}
\[C(n,t,w,w-3) \le
 \frac{\binom{n}{3}}{\binom{w+t}{3}}
 + \frac{\binom{n}{2}}{\binom{w-1}{2}}
 +\frac{n}{w-2}+ 1.\]
Both bounds were proved under the condition
\[2\binom{w}{3}\ge\binom{w-t-1}{3}+\binom{w+t}{3},\]
which, for integers $t\ge1$ and $w\ge t+2$, is equivalent to
\(w\ge 2t(t+1)+2\).
They also determined the asymptotic maximum size of LPECC codes
under the same condition~\cite[Corollary~5]{2025}.

These results suggest two questions. First, does the CPECC upper bound remain valid for weights substantially below the quadratic-order threshold in $t$? Extending its range of validity would allow the optimality of constructions to be established for additional parameters. Second, can the finite-length LPECC bound be strengthened while retaining a matching asymptotic construction under a weaker weight condition?
The distinction is important: the first question concerns the range of validity of an existing bound, whereas the second concerns an improvement to the bound itself. Moreover, the LPECC question requires a separate argument, since information classes may contain transition vectors of different weights.

In this paper, we address these questions for binary CPECC and LPECC codes with $e=w-3$. To state our results, define
\[w_0(t)=\begin{cases}
5, & t=1,\\
7, & t=2,\\
10, & t=3,\\
\left\lceil\dfrac{3+\sqrt{1+8t^2(t+1)}}{2}\right\rceil, & t\ge4,
\end{cases}\]
and let $\mu(t)=w_0(t)$ for $t\ne3$, with $\mu(3)=11$. Our main results are as follows.

\begin{enumerate}
\item[(i)] \textbf{CPECC codes: an improved parameter range and exact and asymptotic constructions.}
For $t\ge1$, $w\ge w_0(t)$, and $n\ge w+t$, we prove
\[C'(n,t,w,w-3)\le\left\lfloor
\frac{\binom{n}{3}}{\binom{w+t}{3}}\right\rfloor.\]
Thus, the upper bound of Zhao and Zhang remains valid under the weaker weight condition $w\ge w_0(t)$.
Using Steiner systems and asymptotically optimal \(3\)-packings, we obtain two families of CPECC codes. Steiner systems $S(3,w+t,n)$ yield codes attaining the bound exactly, while asymptotically optimal $3$-$(n,w+t,1)$ packings yield codes attaining it asymptotically. In particular, for fixed $t$ and $w\ge w_0(t)$,
\[C'(n,t,w,w-3)=(1+o(1))
\frac{\binom{n}{3}}{\binom{w+t}{3}} \quad (n\to\infty). \]

\item[(ii)] \textbf{LPECC codes: a stronger upper bound and an asymptotically optimal construction.}
For $t\ge1$, $w\ge\mu(t)$, and $n\ge w+t$, we establish the new upper bound
\[C(n,t,w,w-3)\le\left\lfloor
\frac{\binom{n+2}{3}}{\binom{w+t}{3}}\right\rfloor.\]
This bound is strictly smaller than the previous bound of Zhao and Zhang whenever both apply, and it holds under a weaker sufficient condition on $w$. The CPECC constructions in (i) are also LPECC codes, and their asymptotic optimality implies that the above bound is asymptotically tight, since \(\binom{n+2}{3}/\binom{n}{3}\to1\) as \(n\to\infty\). Consequently, for fixed $t$ and $w\ge\mu(t)$,
\[C(n,t,w,w-3) = (1+o(1)) \frac{\binom{n}{3}}{\binom{w+t}{3}}
\quad (n\to\infty).\]
\end{enumerate}

Both sufficient weight thresholds satisfy
\[w_0(t)\sim\mu(t)\sim\sqrt{2}\,t^{3/2} \quad (t\to\infty).\]
Thus, for both CPECC and LPECC codes, the previous quadratic-order weight requirement is reduced to order $t^{3/2}$. The CPECC result extends the range in which exact or asymptotic optimality can be established, while the LPECC result additionally improves the finite-length upper bound.
The proofs rely on three-element shadows of set families. In the set-theoretic representation, Property $A(t)$ requires each information class to have transversal number at least $t+1$, while the distance condition for $e=w-3$ forces distinct information classes to have disjoint $3$-shadows. For CPECC codes, we establish the sharp estimate $|\partial_3\mathcal H|\ge\binom{w+t}{3}$ for every nonempty $w$-uniform family $\mathcal H$ with $\tau(\mathcal H)\ge t+1$, provided that $w\ge w_0(t)$. For LPECC codes, we combine two-point lifting and compression with a separate treatment of smaller-weight vectors to obtain the corresponding estimates for lifted mixed-weight families.

The remainder of this paper is organized as follows. Section~\ref{sec:preliminaries} introduces the necessary preliminaries. Section~\ref{sec:CPECC} establishes the uniform $3$-shadow bound and presents the CPECC bounds and constructions. Section~\ref{sec:LPECC-w3} develops the mixed-weight shadow bound and derives the LPECC results. Section~\ref{sec:conclusion} concludes the paper.

\section{Preliminaries}\label{sec:preliminaries}

For a positive integer \(n\), let \([n]=\{1,\ldots,n\}\). For a finite set \(X\), let \(\binom{X}{r}\) denote the family of all \(r\)-subsets of \(X\), and let \(2^X\) denote its power set. We use the convention \(\binom{a}{r}=0\) for \(0\le a<r\). 
A binary vector \(\boldsymbol{x}\) is identified with its support \(\operatorname{Supp}(\boldsymbol{x})=\{i:x_i=1\}\). Its Hamming weight is the size of its support. If \(A\) and \(B\) are the supports of two binary vectors, then their Hamming distance is \(d_H(A,B)=|A\triangle B|=|A|+|B|-2|A\cap B|\).

A family \(\mathcal F\subseteq 2^X\) is called \(k\)-uniform if \(\mathcal F\subseteq\binom{X}{k}\). A \emph{transversal} of \(\mathcal F\) is a set \(T\subseteq X\) such that \(T\cap F\ne\varnothing\) for every \(F\in\mathcal F\). The \emph{transversal number} \(\tau(\mathcal F)\) is the minimum cardinality of a transversal, with \(\tau(\mathcal F)=\infty\) if 
\(\varnothing\in\mathcal F\).
For a nonnegative integer \(r\), the \emph{\(r\)-shadow} of \(\mathcal F\) is 
\(\partial_r\mathcal F=\bigcup_{F\in\mathcal F}\binom{F}{r}\).
Thus, \(\partial_3\mathcal F\) is the family of triples contained in the blocks of \(\mathcal F\). 
If \(J\cap X=\varnothing\), define the family \(J*\mathcal F=\{J\cup F:F\in\mathcal F\}\). For nonempty \(\mathcal F\) and 
\(J\ne\varnothing\), we have \(\tau(J*\mathcal F)=1\). Therefore, transversal conditions are imposed on \(\mathcal F\).

In~\cite{2024}, the following combinatorial characterization of an \((n,t,w, e)\)-LPECC code is given.

\begin{definition}\label{def:code}
Let \(n,t,w,e\) be integers with \(t\ge1\), \(0\le e<w\), and \(n\ge w+t\), and let \(X\) be an \(n\)-element set.
An \((n,t,w,e)\)-LPECC code of size \(M\) consists of pairwise disjoint nonempty information classes 
\(\mathcal P_1,\ldots,\mathcal P_M\subseteq 2^X\) satisfying the following conditions:

\begin{itemize}
     \item \textup{\(A(t)\):} For every \(T\in\binom Xt\) and \(i\in[M]\), there exists \(B\in\mathcal P_i\) such that
           \(B\cap T=\varnothing\).
     \item \textup{\(B(w)\):} Every \(B\in\mathcal P_i\), \(i\in[M]\), satisfies \(|B|\le w\).
     \item \textup{\(C(e)\):} For distinct \(i,j\in[M]\), if \(A\in\mathcal P_i\) and \(B\in\mathcal P_j\), then 
           \(|A\triangle  B|\ge 2e+1\).
\end{itemize}
If every block has size exactly \(w\), the code is called an \((n,t,w,e)\)-CPECC code.
\end{definition}

Property \(A(t)\) holds if and only if no \(t\)-subset of \(X\) is a transversal of \(\mathcal P_i\). Hence, \[\tau(\mathcal P_i)\ge t+1 \quad\text{for every }i\in[M].\]
The maximum sizes of LPECC and CPECC codes are denoted by
\(C(n,t,w,e)\) and \(C'(n,t,w,e)\), respectively. Clearly, every CPECC code is an LPECC code, and hence \(C'(n,t,w,e)\le C(n,t,w,e)\). For \(e=w-3\), the distance condition yields the following structural restrictions.

\begin{lemma}\label{lem:lpecc-structural-w3}
Suppose that an \((n,t,w,w-3)\)-LPECC code has at least two information classes. Then the following hold.
\begin{enumerate}
    \item [\normalfont(i).] If \(A\) and \(B\) belong to distinct information classes, then
        \begin{equation}\label{eq:deficit}
              (w-|A|)+(w-|B|)+2|A\cap B|\le5.
        \end{equation}
    \item [\normalfont(ii).] Every block \(A\) satisfies \(w-5\le |A|\le w\).
    \item [\normalfont(iii).] The \(3\)-shadows of distinct information classes are disjoint.
    \item [\normalfont(iv).] At most one information class contains a block of size at most \(w-3\). Every other information class      contains only blocks of sizes \(w\), \(w-1\), and \(w-2\).
\end{enumerate}
\end{lemma}

\begin{proof}
For blocks \(A\) and \(B\) from distinct information classes, Property \(C(w-3)\) gives
\[|A\triangle B|=|A|+|B|-2|A\cap B|\ge2w-5,\] which is equivalent to~\eqref{eq:deficit}.

For any block \(A\), choose a block \(B\) from a different information class. By Property \(B(w)\), all terms in \eqref{eq:deficit} are nonnegative, and hence \(w-|A|\le5\). Therefore, \(w-5\le |A|\le w\).

If two blocks from distinct information classes contained a common triple, then \(|A\cap B|\ge3\), contradicting \eqref{eq:deficit}. Hence their \(3\)-shadows are disjoint.

Finally, if blocks \(A\) and \(B\) from distinct information classes both had size at most \(w-3\), then \((w-|A|)+(w-|B|)\ge6,\)
contradicting~\eqref{eq:deficit}. Hence at most one information class contains a block of size at most \(w-3\).
\end{proof}

An information class whose blocks all have sizes in \(\{w,w-1,w-2\}\) is called \emph{normal}. If an information class contains a block of size at most \(w-3\), it is called \emph{exceptional}.

Let \(X\) be a \(v\)-element set. An \(r\)-\((v,k,1)\) packing is a family \(\mathcal S\subseteq\binom{X}{k}\) such that every \(r\)-subset of \(X\) is contained in at most one block of \(\mathcal S\). In particular, for \(r=3\), this is equivalent to requiring that any two distinct blocks intersect in at most two points.

\begin{theorem}[R\"odl~\cite{Rodl1985}]\label{thm:packing-existence}
For fixed integers \(k>r\ge2\), there exists an \(r\)-\((v,k,1)\) packing with
\[(1-o(1))\frac{\binom vr}{\binom kr}\] blocks as \(v\to\infty\).
\end{theorem}

If every \(r\)-subset of \(X\) is contained in exactly one block of \(\mathcal S\), then \((X,\mathcal S)\) is an \(r\)-\((v,k,1)\) design, or equivalently a Steiner system \(S(r,k,v)\). In particular, an \(S(3,k,v)\) satisfies
\[|\mathcal S|=\frac{\binom v3}{\binom k3},\] and any two distinct blocks intersect in at most two points.

\begin{theorem}[Keevash~\cite{Keevash2014}]\label{thm:Keevash-design}
For fixed integers \(k>r\ge2\), if \(v\) is sufficiently large and
\[\binom{k-i}{r-i}\Big|\binom{v-i}{r-i},\quad 0\le i\le r-1,\]
then an \(S(r,k,v)\) exists.
\end{theorem}

\section{Uniform 3-shadow bounds and CPECC codes}\label{sec:CPECC}

In this section, we establish a uniform \(3\)-shadow bound for set families with large transversal number and apply it to CPECC codes with \(e=w-3\).
As a consequence, we extend the range of validity of the upper bound for CPECC codes obtained by Zhao and Zhang~\cite{2025} from the previous quadratic-order condition on \(w\) to a weaker condition.

For an \((n,t,w,w-3)\)-CPECC code, the distance requirement implies that any two blocks from distinct information classes intersect in at most two points. Hence, the \(3\)-shadows of distinct information classes are pairwise disjoint. Therefore, bounding the number of information classes reduces to finding a lower bound on the \(3\)-shadow size of a \(w\)-uniform family with transversal number at least \(t+1\).
We first prove a more general \(3\)-shadow estimate involving a fixed common set. This estimate will also be used later in the mixed-weight setting for LPECC codes in Section~\ref{sec:LPECC-w3}. We then derive the uniform \(3\)-shadow bound and apply it to obtain the improved CPECC upper bound and the corresponding optimal and asymptotically optimal constructions.

\subsection{The uniform 3-shadow bound}

\begin{lemma}\label{lem:common-point-shadow}
Let \(t\ge1\), \(m\ge0\), and \(k\ge t+2\), and let \(\varnothing\ne\mathcal F\subseteq\binom Vk\) satisfy \(\tau(\mathcal F)\ge t+1\). Let \(J\) be an \(m\)-set disjoint from \(V\), set \(s=k+m\), and fix \(W\in\mathcal F\).

  \begin{enumerate}
     \item[\normalfont(i).] If \(|W\setminus A|\le t\) for every \(A\in\mathcal F\), then
         \[|\partial_3(J*\mathcal F)|\ge\binom{s+t}{3}.\]

     \item[\normalfont(ii).] If there exists \(A\in\mathcal F\) such that \(|W\setminus A|\ge t+1\) and \(|W\cap A|\ge t,\) then
         \[|\partial_3(J*\mathcal F)|\ge\binom{s+t}{3}+\frac{(s-1)(s-2)-2t^2(t+1)}{2}.\]

     \item[\normalfont(iii).] If \(t\ge2\) and there exists \(A\in\mathcal F\) such that \( |W\setminus A|\ge t+1\) and 
         \(|W\cap A|<t,\) then
         \[|\partial_3(J*\mathcal F)|\ge3\binom{s}{3}-\binom{s-1}{3}-\binom{t+m-1}{3}-\binom{m+1}{3}+\binom{m}{3}.\]

     \item[\normalfont(iv).] If \(t=1\) and there exists \(A\in\mathcal F\) disjoint from \(W\), then
         \[|\partial_3(J*\mathcal F)|\ge2\binom{s}{3}-\binom{m}{3}.\]
\end{enumerate}
\end{lemma}

\begin{proof}
Fix \(W\in\mathcal F\). We distinguish two cases according to the intersection size between \(W\) and a block \(A\in\mathcal F\) avoiding a \(t\)-subset of \(W\).
For (i), fix \(P\in\binom Wt\). Since \(\tau(\mathcal F)\ge t+1\), there exists \(B_P\in\mathcal F\) with 
\(B_P\cap P=\varnothing\). By the assumption in (i), we have \(P\subseteq W\setminus B_P\) and \(|W\setminus B_P|\le t\), which implies \(W\setminus B_P=P\). Since \(|B_P|=|W|=k\), it follows that \(B_P=(W\setminus P)\cup E_P\) for some
\(E_P\in\binom{V\setminus W}{t}\).

We count triples contained in the sets \(B_P\cup J\). For a triple contained in \(B_P\cup J\) but not in \(W\cup J\), let \(a\), \(b\), and \(c\) denote the numbers of its points in \(E_P\), \(J\), and \(W\setminus P\), respectively. Then \(a\ge1\) and \(a+b+c=3\).
For fixed \((a,b,c)\), summing over all \(P\in\binom Wt\), the total number of such triples counted with multiplicity is
\[\binom kt\binom ta\binom mb\binom{k-t}{c}.\]
On the other hand, a fixed triple can be counted at most \(\binom{k-c}{t}\) times. Indeed, \(P\) must avoid its \(c\) points in \(W\), so \(P\) is chosen from the remaining \(k-c\) points of \(W\). Hence, the number of distinct triples of type \((a,b,c)\) is at least
\[\frac{\binom kt\binom ta\binom mb\binom{k-t}{c}}{\binom{k-c}{t}}=\binom ta\binom mb\binom kc.\]
Adding the \(\binom{s}{3}\) triples contained in \(W\cup J\), we obtain
\[
   \begin{aligned}
       |\partial_3(J*\mathcal F)|
      &\ge\sum_{a+b+c=3}\binom ta\binom mb\binom kc=\binom{s+t}{3},
   \end{aligned}
\]
by Vandermonde's identity. This proves (i).

For (ii)--(iv), let \(g_m(x)=\binom{x+m}{3}\) for nonnegative integers \(x\). Since \(g_m(x+1)-g_m(x)=\binom{x+m}{2}\) is nondecreasing, \(g_m\) is discretely convex. In particular, for \(1\le x\le y\),
\[g_m(x-1)+g_m(y+1)\ge g_m(x)+g_m(y).\] Thus, for fixed \(x+y\), the sum \(g_m(x)+g_m(y)\) is maximized when the two arguments are as unequal as the constraints permit.

We now prove (ii). Let \(q=|W\setminus A|\) and \(r=|W\cap A|\). Since \(W\) and \(A\) are both \(k\)-sets, we have
\(q+r=k\) and \(|A\setminus W|=q\). As \(r\ge t\), choose \(T\in\binom{W\cap A}{t}\). By \(\tau(\mathcal F)\ge t+1\), there exists \(B\in\mathcal F\) such that \(B\cap T=\varnothing\).
Let \(\ell=|B\cap W\cap A|\), \(u=|B\cap(W\setminus A)|\), and \(v=|B\cap(A\setminus W)|\).
Then \(\ell\le r-t\). Since these parts are pairwise disjoint and \(|B|=k\), we have \(0\le u,v\le q\) and \(u+v\le k-\ell=q+r-\ell\).
For fixed \(\ell\), the previous inequality shows that \(g_m(\ell+u)+g_m(\ell+v)\) is maximized when \(u+v\) is as large as possible and \(u,v\) are as unequal as the constraints permit. Hence,
     \begin{equation}\label{eq:xxx}
            g_m(\ell+u)+g_m(\ell+v)-g_m(\ell) \le
        \begin{cases}
            2g_m(\ell+q)-g_m(\ell),
            &0\le \ell\le r-q,\\[1mm] g_m(\ell+q)+g_m(r)-g_m(\ell),
            &r-q<\ell\le r-t.
        \end{cases}
     \end{equation}
The first case can occur only when \(r\ge q\), since otherwise \(r-q<0\) and the corresponding range of \(\ell\) is empty. If 
\(\ell\le r-q\), then \(q+r-\ell\ge2q,\)
so the maximum is attained at \(u=v=q\). Otherwise, \(q+r-\ell<2q,\) and the maximum is attained when \(\{u,v\}=\{q,r-\ell\}\). The successive differences of the two expressions in \eqref{eq:xxx} are
\[2\binom{\ell+q+m}{2}-\binom{\ell+m}{2}~~{\rm and}~~\binom{\ell+q+m}{2}-\binom{\ell+m}{2},\] respectively, both of which are nonnegative. Hence, both expressions are increasing in \(\ell\), and the maximum over \(0\le\ell\le r-t\) is attained at \(\ell=r-t\). Therefore,
     \begin{equation}\label{eq:xxx2}
            g_m(\ell+u)+g_m(\ell+v)-g_m(\ell)\le g_m(k-t)+g_m(r)-g_m(r-t).
     \end{equation}
By the definitions of \(r,\ell,u,v\), we have
\[\begin{aligned}
|(J\cup W)\cap(J\cup A)|&=m+r,\\
|(J\cup W)\cap(J\cup B)|&=m+\ell+u,\\
|(J\cup A)\cap(J\cup B)|&=m+\ell+v,\\
|(J\cup W)\cap(J\cup A)\cap(J\cup B)|&=m+\ell.
\end{aligned}\]
Applying the inclusion–exclusion principle to the three families $\binom{J\cup W}{3}$, $\binom{J\cup A}{3}$, and $\binom{J\cup B}{3}$, and using \eqref{eq:xxx2}, we obtain
\[
    \begin{aligned}
          |\partial_3(J*\mathcal F)|
         &\ge 3g_m(k)-g_m(r)-g_m(\ell+u)-g_m(\ell+v)+g_m(\ell)\\
         &\ge 3g_m(k)-g_m(k-t)-2g_m(r)+g_m(r-t).
    \end{aligned}
\]
Let \(h(r)=-2g_m(r)+g_m(r-t)\) for \(r\ge t\). Since \[h(r+1)-h(r)=-2\binom{r+m}{2}+\binom{r-t+m}{2}\le 0,\]
\(h\) is nonincreasing. 
As \(t\le r=k-q\le k-t-1\), we have \(h(r)\ge h(k-t-1)\).
Consequently, 
\[
     \begin{aligned}
            |\partial_3(J*\mathcal F)|
            \ge{}& 3\binom s3-\binom{s-t}{3}-2\binom{s-t-1}{3}+\binom{s-2t-1}{3}.
     \end{aligned}
\]
Since \(k\ge2t+1\), all upper arguments in the above binomial coefficients are nonnegative. A direct expansion gives
\[
     \begin{aligned}
           &3\binom s3-\binom{s-t}{3}-2\binom{s-t-1}{3}+\binom{s-2t-1}{3}-\binom{s+t}{3}\\
           &\qquad=\frac{(s-1)(s-2)-2t^2(t+1)}{2}.
      \end{aligned}
\]
Combining this identity with the preceding lower bound completes the proof of (ii).

   For (iii), assume \(t\ge2\) and let \(r=|W\cap A|<t\). 
Choose a \(t\)-subset \(T\subseteq W\cup A\) containing \(W\cap A\) and meeting both \(W\) and \(A\). By \(\tau(\mathcal F)\ge t+1\), there exists \(B\in\mathcal F\) disjoint from \(T\).
Let \(u=|W\cap B|\) and \(v=|A\cap B|\). As \(T\) meets both \(W\) and \(A\), we have \(u,v\le k-1\).
Moreover, \(W\cap A\subseteq T\) implies \(W\cap A\cap B=\varnothing\), so \(W\cap B\) and
\(A\cap B\) are disjoint subsets of \(B\). Hence \(u+v\le k\). By discrete convexity, we have \(g_m(u)+g_m(v)\le g_m(k-1)+g_m(1)\).
Furthermore, it follows from \(r\le t-1\) that \(g_m(r)\le g_m(t-1)\).
Combining with the inclusion--exclusion principle, we obtain
\[
    \begin{aligned}
           |\partial_3(J*\mathcal F)|
          &\ge 3g_m(k)-g_m(r)-g_m(u)-g_m(v)+g_m(0)\\
          &\ge 3g_m(k)-g_m(k-1)-g_m(t-1)-g_m(1)+g_m(0)\\
          &=3\binom{s}{3}-\binom{s-1}{3}-\binom{t+m-1}{3}-\binom{m+1}{3}+\binom{m}{3}.
    \end{aligned}
\]
This proves (iii).

Finally, suppose that \(t=1\) and \(W\cap A=\varnothing\). Then \((J\cup W)\cap(J\cup A)=J\). By the inclusion--exclusion principle, we have 
\[\begin{aligned}
|\partial_3(J*\mathcal F)|
&\ge\left|\binom{J\cup W}{3}\cup\binom{J\cup A}{3}\right|\\
&=2\binom{s}{3}-\binom{m}{3}.
\end{aligned}\]
This proves (iv).
\end{proof}

For \(t=3\), the preceding lemma yields the desired uniform \(3\)-shadow bound for \(w\ge11\). The following lemma extends this bound to \(w=10\).

\begin{lemma}\label{lem:exceptional-shadow}
Let \(\mathcal H\subseteq\binom{X}{10}\) satisfy \(\tau(\mathcal H)\ge4\). If there exist \(W,A\in\mathcal H\)
such that \(|W\cap A|\le2\), then \[|\partial_3\mathcal H|\ge\binom{13}{3}.\]
\end{lemma}

\begin{proof}
Let \(W,A\in\mathcal H\) satisfy \(r=|W\cap A|\le2\). Suppose first that \(r\in\{1,2\}\). Choose a \(3\)-subset
\(T\subseteq W\cup A\) such that \(W\cap A\subseteq T\) and \(|T\cap W|,\,|T\cap A|\ge2\).
Since \(\tau(\mathcal H)\ge4\), there exists \(B\in\mathcal H\) with \(B\cap T=\varnothing\). Let
\(u=|W\cap B|\) and\(v=|A\cap B|\).
Since \(T\) contains at least two points from each of \(W\) and \(A\), we have \(u,v\le8\). Moreover,
\(W\cap A\subseteq T\) implies \(W\cap A\cap B=\varnothing\). Thus \(W\cap B\) and \(A\cap B\) are disjoint subsets of \(B\), so
\(u+v\le |B|=10\).
By discrete convexity,
\[ 
       \binom{u}{3}+\binom{v}{3}
   \le \binom{8}{3}+\binom{2}{3}.
\]
Since \(r\le2\), we have \(\binom{r}{3}=0\). Hence, by the inclusion--exclusion principle,
\[
    \begin{aligned}
          |\partial_3\mathcal H|
         &\ge 3\binom{10}{3} -\binom{8}{3} -\binom{2}{3}\\
         &=304 >\binom{13}{3}.
    \end{aligned}
\]
It remains to consider \(r=0\), in which case \(W\cap A=\varnothing\). Assume, for a contradiction, that
\[|\partial_3\mathcal H|<\binom{13}{3}=286.\]
Since \(W\) and \(A\) are disjoint, they contribute \(2\binom{10}{3}=240\) distinct triples. Hence all other blocks together can contribute at most \(45\) additional triples.

Moreover, \(\mathcal H\setminus\{W,A\}\ne\varnothing\). Otherwise, one point from each of \(W\) and \(A\) would form a \(2\)-element transversal of \(\mathcal H\), contradicting \(\tau(\mathcal H)\ge4\).
Let \(B\in\mathcal H\setminus\{W,A\}\). Suppose that \(|B\cap W|\le8\) and \(|B\cap A|\le8\).
Since \(W\cap A=\varnothing\),
\(|B\cap W|+|B\cap A|\le |B|=10\).
By discrete convexity,
\[\binom{|B\cap W|}{3}+\binom{|B\cap A|}{3}\le\binom{8}{3}+\binom{2}{3}.\]
Thus \(B\) contributes at least \[\binom{10}{3}-\binom{8}{3}-\binom{2}{3}=64\] triples not contained in \(W\) or \(A\), contradicting the bound above.
Therefore, at least one of \(|B\cap W|\) and \(|B\cap A|\) is at least \(9\). Since \(B\ne W,A\) and all three blocks have size 
\(10\), both intersection sizes are at most \(9\). Hence
\(\max\{|B\cap W|,|B\cap A|\}=9\) for every \(B\in\mathcal H\setminus\{W,A\}\).

Suppose there exist \(B_1,B_2\in\mathcal H\setminus\{W,A\}\) with \(|B_1\cap W|=|B_2\cap A|=9\).
Each contributes \(\binom{9}{2}=36\) additional triples containing two points of \(W\) or \(A\), respectively. Since \(W\cap A=\varnothing\), these collections are disjoint, giving at least \(72\) additional triples, a contradiction.
By symmetry, we may assume that \(|B\cap W|=9\) for every \(B\in\mathcal H\setminus\{W,A\}\).
Choose distinct \(x,y\in W\) and \(z\in A\). Every member of \(\mathcal H\setminus\{W,A\}\) meets
\(\{x,y\}\), so \(\{x,y,z\}\) is a transversal of \(\mathcal H\), contradicting \(\tau(\mathcal H)\ge4\).
\end{proof}

The preceding two lemmas yield the following uniform \(3\)-shadow bound. Define
    \begin{equation}\label{eq:w0}
         w_0(t)=
        \begin{cases}
               5, & t=1,\\
               7, & t=2,\\
               10, & t=3,\\
               \left\lceil\dfrac{3+\sqrt{1+8t^2(t+1)}}{2}\right\rceil, & t\ge4.
        \end{cases}
    \end{equation}

\begin{theorem}\label{thm:three-shadow-transversal}
Let \(t\ge1\) and \(w\ge w_0(t)\). If \(\varnothing\ne\mathcal H\subseteq\binom{X}{w}\) satisfies \(\tau(\mathcal H)\ge t+1\), then \[|\partial_3\mathcal H|\ge\binom{w+t}{3}.\]
Moreover, equality is attained by \(\mathcal H=\binom{S}{w}\) whenever \(S\subseteq X\) and \(|S|=w+t\).
\end{theorem}

\begin{proof}
Apply Lemma~\ref{lem:common-point-shadow} with \(m=0\) and \(k=s=w\), and fix \(W\in\mathcal H\).
If \(|W\setminus A|\le t\) for every \(A\in\mathcal H,\) then it follows from Lemma~\ref{lem:common-point-shadow}(i) that
\(|\partial_3\mathcal H|\ge\binom{w+t}{3}\).
Suppose that there exists \(A\in\mathcal H\) such that \(q=|W\setminus A|\ge t+1, r=|W\cap A|=w-q\).
If \(r\ge t\), then it follows from Lemma~\ref{lem:common-point-shadow}(ii) that
\[|\partial_3\mathcal H|\ge\binom{w+t}{3}+\frac{(w-1)(w-2)-2t^2(t+1)}{2}.\]
Since \(w\ge w_0(t)\), we have
\((w-1)(w-2)\ge2t^2(t+1)\).
Therefore,
\(|\partial_3\mathcal H|\ge\binom{w+t}{3}\).
It remains to consider \(r<t\). If \(t=1\), then \(r=0\), and by Lemma~\ref{lem:common-point-shadow}(iv), we have
\[|\partial_3\mathcal H|\ge2\binom{w}{3}=\binom{w+1}{3}+\frac{w(w-1)(w-5)}{6}.\]
Since \(w\ge w_0(1)=5\), it follows that
\(|\partial_3\mathcal H|\ge\binom{w+1}{3}\).
Assume \(t\ge2\). By Lemma~\ref{lem:common-point-shadow}(iii),
\[|\partial_3\mathcal H|\ge3\binom{w}{3}-\binom{w-1}{3}-\binom{t-1}{3}.\]
Define
    \begin{equation}\label{eq:Dt}
           D_t(w)=3\binom{w}{3}-\binom{w-1}{3}-\binom{t-1}{3}-\binom{w+t}{3}.
    \end{equation}
Thus it remains to prove that \(D_t(w)\ge0\). The successive difference is 
    \begin{equation}\label{eq:Dt-difference}
           D_t(w+1)-D_t(w)=\frac{w^2+(1-2t)w-t^2+t-2}{2}.
    \end{equation}
For \(t=2,3,4,5\), the difference in \eqref{eq:Dt-difference} is positive for \(w\ge7,11,15,19\), respectively. Since
\(D_2(7)=1\), \(D_3(11)=11\), \(D_4(15)=31\), and \(D_5(19)=63\),
we have \(D_t(w)>0\) throughout these ranges.

For \(t\ge6\), the inequality \((w-1)(w-2)\ge2t^2(t+1)\) implies \(w\ge4t\). Indeed, if \(w\le4t-1\), then
\((w-1)(w-2)\le(4t-2)(4t-3)<2t^2(t+1),\) a contradiction. Moreover, \[6D_t(4t)=2t^3+33t^2-41t+12>0.\]
For \(w\ge4t\), by \eqref{eq:Dt-difference}, we have
\[D_t(w+1)-D_t(w)\ge\frac{7t^2+5t-2}{2}>0.\]
Hence \(D_t(w)>0\) throughout this range.
The only remaining case is \((t,w)=(3,10)\) with \(r<3\). Since \(|W\cap A|=r\le2\), Lemma~\ref{lem:exceptional-shadow} gives
\(|\partial_3\mathcal H|\ge\binom{13}{3}\).

Finally, let \(S\subseteq X\) with \(|S|=w+t\), and let \(\mathcal H=\binom{S}{w}\).
Then \(\tau(\mathcal H)=t+1\) and \(\partial_3\mathcal H=\binom{S}{3}\). Hence
\[|\partial_3\mathcal H|=\binom{w+t}{3},\]
and equality is attained.
\end{proof}

\begin{remark}
The thresholds \(w_0(1)=5\) and \(w_0(2)=7\) are sharp. For \(t=1\) and \(w=4\), let \(\mathcal H\) consist of two
disjoint \(4\)-sets. Then \(\tau(\mathcal H)=2\) and \(|\partial_3\mathcal H|=8<\binom{5}{3}\).
For \(t=2\) and \(w=6\), let \(X_1\) and \(X_2\) be disjoint sets with \(|X_1|=6\) and \(|X_2|=7\). Define
\(\mathcal H=\binom{X_1}{6}\cup\binom{X_2}{6}\).
Then \(\tau(\mathcal H)=3\) and
\[|\partial_3\mathcal H|=\binom{6}{3}+\binom{7}{3}=55<\binom{8}{3}.\]
These examples show that the thresholds \(w_0(1)=5\) and
\(w_0(2)=7\) are best possible.
\end{remark}

\subsection{Bounds and constructions for CPECC codes}

We now apply the uniform \(3\)-shadow bound established in the previous subsection to CPECC codes. Since the \(3\)-shadows of distinct information classes are disjoint, this bound immediately yields an improved upper bound for CPECC codes with \(e=w-3\). We then show that this bound is attainable and asymptotically tight by using Steiner systems and asymptotically optimal \(3\)-packings. We first obtain the following upper bound for CPECC codes.

\begin{corollary}\label{cor:CPECC-three-shadow}
Let \(t\ge1\), \(w\ge w_0(t)\), and \(n\ge w+t\). Then
\[C'(n,t,w,w-3)\le\left\lfloor\frac{\binom{n}{3}}{\binom{w+t}{3}}\right\rfloor.\]
\end{corollary}

\begin{proof}
Let \(\mathcal P_1,\ldots,\mathcal P_M\) be the information classes of an \((n,t,w,w-3)\)-CPECC code. By Property \(A(t)\), each \(\mathcal P_i\) satisfies \(\tau(\mathcal P_i)\ge t+1\). Since all blocks have size \(w\), it follows from Theorem~\ref{thm:three-shadow-transversal} that
\[|\partial_3\mathcal P_i|\ge\binom{w+t}{3}\qquad (1\le i\le M).\]
For \(A\in\mathcal P_i\) and \(B\in\mathcal P_j\) with \(i\ne j\), Property \(C(w-3)\) gives
\[|A\triangle B|=2w-2|A\cap B|\ge 2w-5.\]
Thus \(|A\cap B|\le2\), and hence the \(3\)-shadows of distinct information classes are disjoint. Consequently,
\[M\binom{w+t}{3}\le\sum_{i=1}^{M}|\partial_3\mathcal P_i|\le\binom{n}{3}.\]
The result follows.
\end{proof}

\begin{remark}\label{rem:CPECC-threshold-comparison}

Zhao and Zhang~\cite[Theorem~9]{2025} obtained the same upper bound as in Corollary~\ref{cor:CPECC-three-shadow} under the condition
     \begin{equation}\label{eq:old-condition}
           2\binom{w}{3}\ge\binom{w-t-1}{3}+\binom{w+t}{3}.
     \end{equation}
The difference between the two sides is
\[
    \begin{aligned}
          &2\binom{w}{3}-\binom{w-t-1}{3}-\binom{w+t}{3}\\
          &\quad=\frac{w^2-(2t^2+2t+3)w+3t^2+3t+2}{2}.
    \end{aligned}
\]
For integers \(w\ge t+2\), condition~\eqref{eq:old-condition} is equivalent to \(w\ge 2t^2+2t+2\).
In contrast, Corollary~\ref{cor:CPECC-three-shadow} requires only \(w\ge w_0(t)\), where
\[w_0(t)=(\sqrt{2}+o(1))t^{3/2}.\]
Thus the sufficient weight threshold is reduced from quadratic order to order \(t^{3/2}\). Table~\ref{tab:threshold-comparison} compares the two thresholds for \(1\le t\le18\). 

\end{remark}

\begin{table*}[t]
\centering
\caption{Comparison of sufficient weight thresholds for CPECC codes.}
\label{tab:threshold-comparison}
\renewcommand{\arraystretch}{1.15}
\setlength{\tabcolsep}{3pt}
\small
\begin{tabular}{c|*{18}{c}}
\hline \(t\)
& 1 & 2 & 3 & 4 & 5 & 6 & 7 & 8 & 9
& 10 & 11 & 12 & 13 & 14 & 15 & 16 & 17 & 18 \\
\hline Zhao--Zhang~\cite{2025}
& 6 & 14 & 26 & 42 & 62 & 86 & 114 & 146 & 182
& 222 & 266 & 314 & 366 & 422 & 482 & 546 & 614 & 686 \\
Our result
& 5 & 7 & 10 & 15 & 19 & 24 & 30 & 36 & 42
& 49 & 56 & 63 & 71 & 79 & 87 & 95 & 104 & 113 \\
\hline
\end{tabular}
\end{table*}

Steiner systems provide the equality case for the improved bound.

\begin{corollary}\label{cor:CPECC-Steiner}
Let \(t\ge1\) and \(w\ge w_0(t)\). If an \(S(3,w+t,n)\) exists, then \[ C'(n,t,w,w-3) = \frac{\binom n3}{\binom{w+t}{3}}. \] 
\end{corollary} 

\begin{proof}
Let \((X,\mathcal S)\) be an \(S(3,w+t,n)\). For each \(A\in\mathcal S\), define \(\mathcal P_A=\binom{A}{w}\).
Since \(|A|=w+t\), we have \(\tau(\mathcal P_A)=t+1\), and every block has size \(w\).
Distinct blocks of \(\mathcal S\) intersect in at most two points. Hence any two blocks from distinct information classes have
symmetric difference at least \(2w-4\ge2(w-3)+1\). Thus these information classes form an \((n,t,w,w-3)\)-CPECC code.
Since every triple of \(X\) belongs to exactly one block of \(\mathcal S\), we have
\[|\mathcal S|=\frac{\binom n3}{\binom{w+t}3}.\]
The construction therefore attains the upper bound in Corollary~\ref{cor:CPECC-three-shadow}.
\end{proof}

By Keevash's existence theorem, such Steiner systems exist for all sufficiently large admissible lengths.

\begin{corollary}\label{cor:CPECC-large-n}
Let \(t\ge1\) and \(w\ge w_0(t)\) be fixed. For all sufficiently large \(n\) satisfying \(\binom{w+t-i}{3-i} \mid \binom{n-i}{3-i}, 0\le i\le2,\) we have
\[C'(n,t,w,w-3)=\frac{\binom{n}{3}}{\binom{w+t}{3}}.\]
\end{corollary}

\begin{proof}
By Theorem~\ref{thm:Keevash-design}, an \(S(3,w+t,n)\) exists for all sufficiently large \(n\) satisfying the stated divisibility conditions.
The result follows from Corollary~\ref{cor:CPECC-Steiner}.
\end{proof}

Using asymptotically optimal \(3\)-packings, we obtain the following asymptotic result without divisibility conditions. 

\begin{corollary}\label{cor:CPECC-asymptotic-w3}
Let \(t\ge1\) and \(w\ge w_0(t)\) be fixed. Then
    \begin{equation}\label{eq:CPECC-asymptotic}
            C'(n,t,w,w-3)=(1+o(1))\frac{\binom{n}{3}}{\binom{w+t}{3}}\qquad (n\to\infty).
    \end{equation}
\end{corollary}

\begin{proof}
By Theorem~\ref{thm:packing-existence}, there exists a \(3\)-\((n,w+t,1)\) packing \(\mathcal S\) with
\[|\mathcal S|=(1-o(1))\frac{\binom{n}{3}}{\binom{w+t}{3}}.\]
For each \(A\in\mathcal S\), let \(\mathcal P_A=\binom{A}{w}\). As in the proof of Corollary~\ref{cor:CPECC-Steiner}, these information classes form an \((n,t,w,w-3)\)-CPECC code. 
It follows that
\[C'(n,t,w,w-3)\ge(1-o(1))\frac{\binom{n}{3}}{\binom{w+t}{3}}.\]
Together with the upper bound in Corollary~\ref{cor:CPECC-three-shadow}, this proves \eqref{eq:CPECC-asymptotic}.
\end{proof}

\section{Nonuniform 3-shadow bounds and LPECC codes}\label{sec:LPECC-w3}

In this section, we establish a nonuniform \(3\)-shadow bound and apply it to obtain an improved upper bound for LPECC codes with \(e=w-3\). The obtained bound is strictly stronger than the previous bound of Zhao and Zhang~\cite{2025} whenever both bounds apply.

Unlike CPECC codes, LPECC codes may contain blocks of different sizes, which leads to a mixed-weight shadow problem. Recall that an information class is called normal if all its blocks have sizes in \(\{w,w-1,w-2\}\), and exceptional otherwise. By Lemma~\ref{lem:lpecc-structural-w3}, at most one information class can be exceptional. We first establish the \(3\)-shadow bound for normal information classes and then handle the possible exceptional class. Combining these estimates yields the improved LPECC upper bound.

\subsection{The nonuniform 3-shadow bound}

Let \[\mathcal H=\mathcal H_w\cup\mathcal H_{w-1}\cup\mathcal H_{w-2}, \qquad \mathcal H_j\subseteq\binom{V}{j},\]
and let \(\alpha,\beta\) be two distinct points outside \(V\). Define \(U=V\cup\{\alpha,\beta\}\) and the lifted \(w\)-uniform family \[\widehat{\mathcal H}=\mathcal H_w\cup\{H\cup\{\alpha\}:H\in\mathcal H_{w-1}\}\cup\{H\cup\{\alpha,\beta\}:H\in\mathcal H_{w-2}\}\subseteq\binom Uw.\]
The corresponding lifted \(3\)-shadow size is denoted by
\[\Psi_3(\mathcal H)=|\partial_3\widehat{\mathcal H}|.\]
Since lifting does not preserve the transversal number in general,
Theorem~\ref{thm:three-shadow-transversal} cannot be applied directly to \(\widehat{\mathcal H}\). We therefore introduce a compression operation that transforms the family into one whose members contain the same two points without increasing its \(3\)-shadow size.

For a \(k\)-uniform family \(\mathcal A\subseteq\binom Uk\) and
distinct \(a,x\in U\), define the compression operation \(C_{a\leftarrow x}\) as follows. For each member \(A\in\mathcal A\) satisfying \(x\in A\) and \(a\notin A\), replace \(A\) by \((A\setminus\{x\})\cup\{a\}\), while leaving all other members unchanged. After removing repetitions, the resulting family
is denoted by \(C_{a\leftarrow x}(\mathcal A)\).

\begin{lemma}\label{lem:compression}
Let \(U\) be a finite set, let \(\mathcal A\subseteq\binom Uk\), and let \(a,x\in U\) be distinct. Then, for every \(1\le r\le k\), we have
\[|\partial_r C_{a\leftarrow x}(\mathcal A)|\le|\partial_r\mathcal A|.\]
\end{lemma}

\begin{proof}
Let \(\mathcal C=C_{a\leftarrow x}(\mathcal A)\).  For any \(R\in\partial_r\mathcal C\setminus\partial_r\mathcal A\),
we have \(a\in R\) and \(x\notin R\). Define
\[R'=(R\setminus\{a\})\cup\{x\}.\]
Then \(R'\in\partial_r\mathcal A\). We claim that \(R'\notin\partial_r\mathcal C\). Otherwise, there exists
\(C\in\mathcal C\) such that \(R'\subseteq C\). Since \(x\in R'\), we have \(x\in C\). Thus, \(C\) must be an unchanged
member of \(\mathcal A\) under the compression, and hence \(a\in C\). Therefore, we have
\[R=(R'\setminus\{x\})\cup\{a\}\subseteq C,\]
which implies \(R\in\partial_r\mathcal A\), a contradiction.
Thus \(R'\in\partial_r\mathcal A\setminus \partial_r\mathcal C\). The map \(R\mapsto R'\) is injective, since \(R\) can be uniquely recovered from \(R'\). Therefore, we have
\[|\partial_r\mathcal C\setminus\partial_r\mathcal A|\le |\partial_r\mathcal A\setminus\partial_r\mathcal C|,\]
which implies
\(|\partial_r\mathcal C|\le|\partial_r\mathcal A|\).
This completes the proof.
\end{proof}

Repeated applications of the above compression operation transform the lifted family \(\widehat{\mathcal H}\) into a family of the form \(J*\mathcal F\) without increasing the \(3\)-shadow. After removing the common points in \(J\), the resulting \((w-2)\)-uniform family \(\mathcal F\) satisfies \(\tau(\mathcal F)\ge\tau(\mathcal H)\). The following lemma gives this reduction.

\begin{lemma}\label{lem:mixed-reduction}
Let $V$ be a finite set, let $w\ge5$, and let $J=\{\alpha,\beta\}$ be disjoint from $V$. For every nonempty family \(\mathcal H\subseteq\binom Vw\cup\binom V{w-1}\cup\binom V{w-2}\), there exists a nonempty family $\mathcal F\subseteq\binom V{w-2}$ such that
\[|\partial_3(J*\mathcal F)|\le \Psi_3(\mathcal H)~~{\rm and} ~~\tau(\mathcal F)\ge\tau(\mathcal H).\]
\end{lemma}

\begin{proof}
Let \(U=V\cup J\). Starting from \(\widehat{\mathcal H}\subseteq\binom Uw\), apply the compressions \(C_{\alpha\leftarrow x}\) and \(C_{\beta\leftarrow x}\) for \(x\in V\) until every member contains both \(\alpha\) and \(\beta\). The resulting family has the form \(J*\mathcal F\) for some nonempty \(\mathcal F\subseteq\binom{V}{w-2}\). Since the \(3\)-shadow does not increase under compression by Lemma~\ref{lem:compression},
\[|\partial_3(J*\mathcal F)|\le |\partial_3\widehat{\mathcal H}|=\Psi_3(\mathcal H).\]
For each \(H\in\mathcal H\), let \(F_H\in\mathcal F\) be the member
obtained from \(H\) after compression and deletion of \(J\). Since \(F_H\subseteq H\), every transversal of \(\mathcal F\) is also a transversal of \(\mathcal H\). Hence
\(\tau(\mathcal F)\ge\tau(\mathcal H)\).
\end{proof}

The preceding lemma reduces the problem to families of the form \(J*\mathcal F\). We first establish the bounds for the two smallest values of \(t\), namely \(t=1\) and \(t=2\).

\begin{proposition}\label{prop:base-t1}
Let \(J\) be a two-element set disjoint from \(V\). If $\varnothing\ne\mathcal F\subseteq\binom{V}{k}$, $k\ge3$, and $\tau(\mathcal F)\ge2$, then $|\partial_3(J*\mathcal F)|\ge\binom{k+3}{3}$.
\end{proposition}

\begin{proof}
Fix \(W\in\mathcal F\). If \(|W\setminus A|\le1 \) for all \(A\in\mathcal F,\) then Lemma~\ref{lem:common-point-shadow}(i) with \((t,m)=(1,2)\) gives the result. Otherwise, choose \(A\in\mathcal F\) with \(|W\setminus A|\ge2\). If \(W\cap A\ne\varnothing\), then Lemma~\ref{lem:common-point-shadow}(ii) with \(s=k+2\) gives
\[|\partial_3(J*\mathcal F)| \ge \binom{k+3}{3}+\frac{k(k+1)-4}{2}\ge \binom{k+3}{3}.\]
If \(W\cap A=\varnothing\), Lemma~\ref{lem:common-point-shadow}(iv) gives
\[|\partial_3(J*\mathcal F)| \ge 2\binom{k+2}{3} \ge \binom{k+3}{3}.\]
This completes the proof.
\end{proof}

\begin{proposition}\label{prop:base-t2}
Let \(J\) be a two-element set disjoint from \(V\). If $\varnothing\ne\mathcal F\subseteq\binom{V}{k}$, $k\ge5$, and $\tau(\mathcal F)\ge3$, then $|\partial_3(J*\mathcal F)|\ge\binom{k+4}{3}$.
\end{proposition}

\begin{proof}
Fix $W\in\mathcal F$. If $|W\setminus A|\le2$ for every $A\in\mathcal F$, Lemma~\ref{lem:common-point-shadow}(i) applies. Otherwise choose $A$ with $|W\setminus A|\ge3$. If $|W\cap A|\ge2$, part (ii), with $s=k+2$, gives
\[|\partial_3(J*\mathcal F)|\ge\binom{k+4}{3}+\frac{k(k+1)-24}{2} \ge\binom{k+4}{3}.\]
It remains to consider \(|W\cap A|\le1\). By Lemma~\ref{lem:common-point-shadow}(iii),
\[|\partial_3(J*\mathcal F)|\ge3\binom{k+2}{3}-\binom{k+1}{3}-2.\]
The difference between the right-hand side and $\binom{k+4}{3}$ is
\[f(k)=\frac{k^3-19k-36}{6}.\]
Since $f(6)=11$ and $f(k+1)-f(k)=(k+3)(k-2)/2>0$ for $k\ge6$, we have \(f(k)>0\) for all \(k\ge6\). Hence it remains only to
consider the case \(k=5\).

Let \(J=\{\alpha,\beta\}\), \(\widetilde W=J\cup W\), and \(\widetilde A=J\cup A\). Choose a \(2\)-subset
\(T\subseteq W\cup A\) containing \(W\cap A\) and meeting both blocks. By \(\tau(\mathcal F)\ge 3\) and \(|T|=2\), there exists
\(B\in\mathcal F\) disjoint from \(T\). Let \(\widetilde B=J\cup B\). Then \(|\widetilde W\cap \widetilde A|\le3\) and \(|\widetilde W\cap \widetilde A\cap \widetilde B|=2\).
Let $u=|\widetilde W\cap \widetilde B|$ and $v=|\widetilde A\cap \widetilde B|$. Then $u,v\le6$ and $u+v\le9$. Discrete convexity gives $\binom{u}{3}+\binom{v}{3}\le\binom{6}{3}+\binom{3}{3}=21$. Therefore, 
\[\left|\binom{\widetilde W}{3}\cup\binom{\widetilde A}{3}\cup\binom{\widetilde B}{3}\right|
\ge3\binom{7}{3}-1-21=83.\]
Suppose that $|\partial_3(J*\mathcal F)|\le83$. Then equality holds throughout the above estimates. In particular,
\(|W\cap A|=1\) and \(\{u,v\}=\{3,6\}\).
Let \(x\) denote the unique element of \(W\cap A\). By symmetry, we
may assume that
\[
    \begin{aligned}
          \widetilde W&=J\cup\{x\}\cup P,\\
          \widetilde A&=J\cup\{x\}\cup Q,\\
          \widetilde B&=J\cup Q\cup\{p_0\},
    \end{aligned}
\]
where \(P,Q\) are disjoint \(4\)-sets and \(p_0\in P\).
Let \(\widetilde D=J\cup F\) be any member of \(J*\mathcal F\).
Since equality holds,
\[\binom{\widetilde D}{3}\subseteq\binom{\widetilde W}{3}\cup\binom{\widetilde A}{3}\cup\binom{\widetilde B}{3}.\]
In particular,
\(\widetilde D\subseteq\widetilde W\cup\widetilde A\cup\widetilde B\).

If \(\widetilde D\cap(P\setminus\{p_0\})\ne\varnothing\), choose \(z\) in this intersection. Since \(z\) belongs to neither \(\widetilde A\) nor \(\widetilde B\), every triple of \(\widetilde D\) containing \(z\) must belong to \(\binom{\widetilde W}{3}\). Thus \(\{\alpha,z,y\}\subseteq\widetilde W\) for every \(y\in F\setminus\{z\}\), which implies \(\widetilde D\subseteq\widetilde W\). Therefore \(\widetilde D=\widetilde W\).
Otherwise, \(\widetilde D\cap(P\setminus\{p_0\})=\varnothing\), and hence \(\widetilde D\subseteq J\cup\{x,p_0\}\cup Q\).
Since \(|\widetilde D|=7\), if \(\widetilde D\notin\{\widetilde A,\widetilde B\}\), then \(\widetilde D\) contains \(x,p_0\) and some \(q\in Q\). However, the triple \(\{x,p_0,q\}\) belongs to none of \(\widetilde W,\widetilde A,\widetilde B\), contradicting the shadow inclusion. Hence \(\widetilde D\in\{\widetilde A,\widetilde B\}\).

Therefore every member of \(J*\mathcal F\) is one of \(\widetilde W,\widetilde A,\widetilde B\). Consequently, \(\{x,p_0\}\) meets every member of \(\mathcal F\), contradicting \(\tau(\mathcal F)\ge3\). Hence
\[|\partial_3(J*\mathcal F)|\ge84=\binom93 .\]
This proves the case where \(k=5\).
\end{proof}

The preceding propositions, together with Lemma~\ref{lem:common-point-shadow}, yield the following mixed-weight \(3\)-shadow bound. Define
\[
   \mu(t)=
        \begin{cases}
            11, & t=3,\\
            w_0(t), & t\ne3.
        \end{cases}
\]

\begin{theorem}\label{thm:mixed-three-shadow}
Let \(t\ge1\) and \(w\ge\mu(t)\). If \(\varnothing\ne\mathcal H\subseteq \binom{V}{w}\cup\binom{V}{w-1}\cup\binom{V}{w-2}\)
satisfies \(\tau(\mathcal H)\ge t+1\), then
    \[\Psi_3(\mathcal H)\ge\binom{w+t}{3}.\]
Moreover, equality is attained by \(\mathcal H=\binom{S}{w}\) whenever \(S\subseteq V\) and \(|S|=w+t\).
\end{theorem}

\begin{proof}
By Lemma~\ref{lem:mixed-reduction}, it suffices to prove \(|\partial_3(J*\mathcal F)|\ge\binom{w+t}{3}\) for every nonempty \(\mathcal F\subseteq\binom{V}{w-2}\) with \(\tau(\mathcal F)\ge t+1\). The cases \(t=1,2\) follow from Propositions~\ref{prop:base-t1} and~\ref{prop:base-t2}.

Let \(t\ge3\) and \(k=w-2\). By the definition of \(\mu(t)\), \(k\ge t+2\) and \((w-1)(w-2)\ge2t^2(t+1)\).
Fix \(W\in\mathcal F\). If \(|W\setminus A|\le t\) for every \(A\in\mathcal F\), Lemma~\ref{lem:common-point-shadow}(i) with \(m=2\) gives the desired bound.
Otherwise, choose \(A\in\mathcal F\) with \(|W\setminus A|\ge t+1\). If \(|W\cap A|\ge t\), Lemma~\ref{lem:common-point-shadow}(ii) gives
\[|\partial_3(J*\mathcal F)| \ge\binom{w+t}{3} +\frac{(w-1)(w-2)-2t^2(t+1)}{2} \ge\binom{w+t}{3}.\]
If \(|W\cap A|<t\), Lemma~\ref{lem:common-point-shadow}(iii) gives
\[|\partial_3(J*\mathcal F)|\ge3\binom{w}{3}-\binom{w-1}{3}-\binom{t+1}{3}-1.\]
Since \(w\ge\mu(t)\ge4t-1\), the difference between the right-hand side and \(\binom{w+t}{3}\) is increasing
in \(w\).  Indeed, its successive difference is
\[3\binom{w}{2}-\binom{w-1}{2}-\binom{w+t}{2}=\frac{w^2-(2t-1)w-t^2+t-2}{2}>0\]
for \(w\ge4t-1\). At \(w=4t-1\), this difference equals
\[
    \begin{aligned}
          &3\binom{4t-1}{3}-\binom{4t-2}{3}-\binom{t+1}{3}-1-\binom{5t-1}{3}\\
          &\qquad=\frac{(t-3)(t+1)(t+5)}{3}+6>0.
    \end{aligned}
\]
Therefore, the desired inequality follows.

Finally, choose \(S\subseteq V\) with \(|S|=w+t\), and let \(\mathcal H=\binom{S}{w}\). Then
\[\tau(\mathcal H)=t+1,\quad \Psi_3(\mathcal H)=\binom{w+t}{3}.\]
Thus, equality is attained.
\end{proof}

Theorem~\ref{thm:mixed-three-shadow} gives a lower bound for the lifted \(3\)-shadow of each normal information class. We use two fixed points \(\alpha,\beta\notin X\) to lift all normal information classes. The following lemma shows that, if an exceptional class exists, at least \(\binom{w+t}{3}\) triples remain outside the union of their lifted \(3\)-shadows.

\begin{lemma}\label{lem:exceptional-reserve}
Let \(t\ge1\), \(w\ge\mu(t)\), and \(n\ge w+t\). Let \(X\) be an \(n\)-element set, and let \(\alpha,\beta\notin X\) be distinct
points. If an \((n,t,w,w-3)\)-LPECC code on \(X\) has at least two information classes, including an exceptional class, then at least \(\binom{w+t}{3}\) triples in \(\binom{X\cup\{\alpha,\beta\}}{3}\) are lie outside the union of the lifted \(3\)-shadows of the normal classes.
\end{lemma}

\begin{proof}
Choose \(A\) from the exceptional information class with \(|A|\le w-3\), and let \(a=w-|A|\). By Lemma~\ref{lem:lpecc-structural-w3}, \(a\in\{3,4,5\}\). For any member \(B\) of a normal information class, let
\[b=w-|B|\in\{0,1,2\}, \quad r=|A\cap B|.\]
The same lemma gives
    \begin{equation}\label{eq:exceptional-deficit}
           a+b+2r\le5.
    \end{equation}
 Let \(J=\{\alpha,\beta\}\). We consider the possible values of \(a\).   
 \begin{itemize}
        \item[\normalfont(i).] Suppose that \(a=3\). By~\eqref{eq:exceptional-deficit}, \(r\le1\) if \(b=0\), while \(r=0\) if \(b=1,2\). Thus a normal lifted member contains at most one point of \(A\), and it cannot contain both a point of \(A\) and a point of \(J\). Consequently, every triple containing at least two points of \(A\), or exactly one point of \(A\) and at least one point of \(J\), is uncovered by the normal lifted \(3\)-shadows. Hence the number of such triples in \(X\cup J\) is
              \[\begin{aligned}
                      R(n,w)={}&\binom{w-3}{3}+\binom{w-3}{2}(n-w+5)\\
                      &+(w-3)\bigl(2(n-w+3)+1\bigr)\\
                      ={}&\frac{w-3}{6}\bigl(3nw-2w^2+6w+2\bigr)
              \end{aligned}\]

        \item[\normalfont(ii).] Suppose that \(a=4\). Then \eqref{eq:exceptional-deficit} implies \(b\le1\) and \(r=0\). Hence every normal lifted member is disjoint from \(A\cup\{\beta\}\). If \(a=5\), then \(b=r=0\), and every normal lifted member is disjoint from \(A\cup J\). In either case, there is a fixed \((w-3)\)-set disjoint from every normal lifted member. Hence at least Thus at least \( \binom{n+2}{3}-\binom{n-w+5}{3}\) triples are uncovered.
\end{itemize}   
Since
\[\binom{n+2}{3}-\binom{n-w+5}{3}-R(n,w)=(w-3)\binom{n-w+3}{2}\ge0,\]
there are at least \(R(n,w)\) uncovered triples in each case.
It remains to prove that \(R(n,w)\ge\binom{w+t}{3}\). The distance condition and the existence of at least two information classes imply \(n\ge2w-5\). Together with \(n\ge w+t\), this yields \(n\ge\max\{2w-5,w+t\}\).
As \(R(n,w)\) is increasing in \(n\), it suffices to consider \(n=\max\{2w-5,w+t\}\).

First suppose \(t\ge3\).We claim that \(\mu(t)\ge2t+5\). For \(t=3\), we have \(\mu(3)=11=2t+5\). For \(t\ge4\), by the definition of \(\mu(t)\), \(\mu(t)=\left\lceil(3+\sqrt{1+8t^2(t+1)})/2\right\rceil\ge2t+5\). Hence \(2w-5\ge w+t\), and it suffices to consider \(n=2w-5\).
Define
\[E(w,t)=6\left(R(2w-5,w)-\binom{w+t}{3}\right).\]
For fixed \(t\), viewing \(E(w,t)\) as a polynomial in \(w\), we have
\[\frac{1}{3}\frac{\partial E}{\partial w}=3w^2-2(t+6)w-t^2+2t+9.\]
The right-hand side is increasing for \(w\ge2t+5\) and is positive at
\(w=2t+5\). Hence \(E(w,t)\) is increasing there. Since
\[E(2t+5,t)=(t+1)(5t^2+43t+54)>0,\]
we obtain \(R(n,w)\ge\binom{w+t}{3}\) for \(t\ge3\).

For \(t=1\) and \(w=5\), we have \(n\ge6\), and \(R(6,5)=24>\binom63\).
For \(t=1\), \(w\ge6\), and for \(t=2\), \(w\ge7\), we have \(2w-5\ge w+t\). Thus it suffices to consider \(n=2w-5\). At the respective starting points,
\[
    \begin{aligned}
           E(6,1)&=66, & \frac{1}{3}E_w(6,1)&=34,\\
           E(7,2)&=36, &\frac{1}{3}E_w(7,2)&=44,
    \end{aligned}
\]
where \(E_w=\partial E/\partial w\). By the expression for \(E_w\) above, it remains positive as \(w\) increases in both cases. Hence \(E(w,t)>0\), proving \(R(n,w)\ge\binom{w+t}{3}\) for \(t=1,2\).

\end{proof}

\subsection{Upper Bounds for LPECC Codes}

We now apply the nonuniform \(3\)-shadow bound established in the previous subsection to LPECC codes. Since LPECC codes may contain transition vectors of different weights, we need to take into account the possible existence of an exceptional information class in addition to the normal information classes. Combining the shadow bound for normal classes with the contribution of the exceptional class, we obtain a new upper bound for LPECC codes with \(e=w-3\), which is strictly stronger than the previous bound of Zhao and Zhang~\cite{2025} whenever both bounds apply.
Combining Theorem~\ref{thm:mixed-three-shadow} with Lemma~\ref{lem:exceptional-reserve}, we obtain the following upper bound.

\begin{theorem}\label{thm:LPECC-upper-w3}
Let $t\ge1$, $w\ge\mu(t)$, and $n\ge w+t$. Then \begin{equation}\label{eq:LPECC-upper-w3}
C(n,t,w,w-3)\le\left\lfloor\frac{\binom{n+2}{3}}{\binom{w+t}{3}}\right\rfloor .
\end{equation}
\end{theorem}

\begin{proof}
Let \(\mathcal P_1,\ldots,\mathcal P_M\subseteq 2^X\), where \(|X|=n\), be the information classes of the code. 
By Property \(A(t)\), \(\tau(\mathcal P_i)\ge t+1\) for \(1\le i\le M\).
Let \(L=\binom{w+t}{3}\). Since \(n\ge w+t\), the desired bound is immediate when \(M=1\). Hence assume that \(M\ge2\).
By Lemma~\ref{lem:lpecc-structural-w3}, at most one information class contains a member of size at most \(w-3\), and all
remaining classes are normal. Lift every normal class using the same pair of new points \(\alpha,\beta\). By Theorem~\ref{thm:mixed-three-shadow}, we have
\(|\partial_3\widehat{\mathcal P}_i|\ge L\)
for each normal information class \(\mathcal P_i\).

We claim that the lifted \(3\)-shadows of distinct normal information classes are disjoint. Let \(A\in\mathcal P_i\) and
\(B\in\mathcal P_j\), where \(i\ne j\), and let
\[a=w-|A|,\quad b=w-|B|,\quad r=|A\cap B|.\]
Then \(a,b\in\{0,1,2\}\), and~\eqref{eq:deficit} gives \(a+b+2r\le5\). By symmetry, assume \(a\le b\).
Exactly \(a\) of the newly added points are common to the two lifted sets. Therefore, 
\[2|\widehat A\cap\widehat B|=2(r+a)\le5+a-b\le5,\]
which implies \(|\widehat A\cap\widehat B|\le2\).

If all \(M\) information classes are normal, their pairwise disjoint lifted \(3\)-shadows give \(ML\le\binom{n+2}{3}\).
If an exceptional class exists, there are \(M-1\) normal classes. By Lemma~\ref{lem:exceptional-reserve}, at least \(L\)
triples lie outside the union of their lifted \(3\)-shadows.
Hence
\[(M-1)L\le\binom{n+2}{3}-L,\]
which yields the same bound \(ML\le\binom{n+2}{3}\). Therefore,
\[M\le \left\lfloor \frac{\binom{n+2}{3}}{\binom{w+t}{3}}\right\rfloor.\]
This completes the proof.
\end{proof}

\begin{remark}\label{rem:LPECC-upper-comparison}
For binary LPECC codes with \(e=w-3\), Zhao and Zhang~\cite[Theorem~7]{2025} give
\[
    \begin{aligned}
           C(n,t,w,w-3)\le \frac{\binom{n}{3}}{\binom{w+t}{3}} +1+\frac{n}{w-2}+\frac{\binom{n}{2}}{\binom{w-1}{2}}
    \end{aligned}
\]
under the condition \(2\binom{w}{3}\ge\binom{w-t-1}{3}+\binom{w+t}{3},\) which is equivalent to \(w\ge2t(t+1)+2\).
In comparison, Theorem~\ref{thm:LPECC-upper-w3} gives
\[C(n,t,w,w-3)\le\left\lfloor\frac{\binom{n+2}{3}}{\binom{w+t}{3}}\right\rfloor\]
for $w\ge\mu(t)$ and $n\ge w+t$. Whenever both bounds apply, the latter is strictly smaller. Indeed,
\[
\frac{\binom{n+2}{3}-\binom n3}{\binom{w+t}{3}}=\frac{n^2}{\binom{w+t}{3}}
<\frac{n^2}{(w-1)(w-2)}<\frac{n}{w-2}+\frac{\binom n2}{\binom{w-1}{2}},
\]
where we used \(\binom{w+t}{3}>(w-1)(w-2)\) for \(t\ge1\) and \(w\ge5\).
Moreover, the sufficient condition on $w$ is substantially relaxed. For $t\ge4$,
\[\mu(t)=\left\lceil\frac{3+\sqrt{1+8t^2(t+1)}}{2}\right\rceil=(\sqrt{2}+o(1))t^{3/2}\]
as $t\to\infty$, whereas the condition of Zhao and Zhang is equivalent to $w\ge2t(t+1)+2$.
\end{remark}

Combining Theorem~\ref{thm:LPECC-upper-w3} with the known asymptotic lower bounds gives the following corollary.

\begin{corollary}\label{cor:LPECC-asymptotic-w3}
Let \(t\ge1\) and \(w\ge\mu(t)\) be fixed. Then \begin{equation}\label{eq:LPECC-asymptotic}
C(n,t,w,w-3) =(1+o(1))\frac{\binom{n}{3}}{\binom{w+t}{3}}
\qquad (n\to\infty).
\end{equation}
\end{corollary}

\begin{proof}
By Corollary~\ref{cor:CPECC-asymptotic-w3} and \(C(n,t,w,w-3)\ge C'(n,t,w,w-3)\),
\[(1-o(1))\frac{\binom n3}{\binom{w+t}{3}}\le C(n,t,w,w-3)\le\frac{\binom{n+2}{3}}{\binom{w+t}{3}} .
\]
Since\(\binom{n+2}{3}/\binom{n}{3}\to1\) as \(n\to\infty\), the result follows.


\end{proof}

\begin{remark}\label{rem:LPECC-asymptotic-comparison}
Zhao and Zhang~\cite[Corollary~5]{2025} obtained the same asymptotic formula under condition~\eqref{eq:old-condition}. Corollary~
\ref{cor:LPECC-asymptotic-w3} extends this asymptotic result to the range \(w\ge\mu(t)\).
\end{remark}

\section{Conclusion}\label{sec:conclusion}

In this paper, we studied binary CPECC and LPECC codes with
\(e=w-3\). For CPECC codes, we extended the range of validity of the upper bound previously obtained by Zhao and Zhang~\cite{2025} from the quadratic-order condition \(w\ge 2t(t+1)+2\) to the weaker condition \(w\ge w_0(t)\sim \sqrt{2}\,t^{3/2}\).
Using Steiner systems, we showed that this bound is attained whenever an \(S(3,w+t,n)\) exists and is asymptotically tight for fixed \(t\) and \(w\).
For LPECC codes, we established a new upper bound
\[C(n,t,w,w-3)\le\left\lfloor
\frac{\binom{n+2}{3}}{\binom{w+t}{3}}\right\rfloor\]
for \(w\ge\mu(t)\), where \(\mu(t)\sim \sqrt{2}\,t^{3/2}\).
This bound is strictly stronger than the previous bound of Zhao and Zhang whenever both bounds apply. Moreover, the asymptotic tightness of the CPECC constructions implies that the new LPECC bound is also asymptotically tight.

The proofs of these results are based on sharp uniform and mixed-weight \(3\)-shadow bounds. In particular, the uniform \(3\)-shadow bound gives the improved CPECC result, while the nonuniform \(3\)-shadow bound obtained through lifting and compression handles the mixed-weight structure of LPECC codes.
Consequently, for fixed \(t\) and \(w\) in the respective parameter ranges, the maximum sizes of both CPECC and LPECC codes satisfy
\[C'(n,t,w,w-3)\sim C(n,t,w,w-3)\sim\frac{\binom{n}{3}}{\binom{w+t}{3}}\quad (n\to\infty).\]
Thus, our results reduce the previous quadratic-order requirement on the weight parameter to order \(t^{3/2}\), and determine the asymptotic sizes of CPECC and LPECC codes for \(e=w-3\) in the considered parameter ranges.\\



\end{sloppypar}
\end{document}